\documentclass[aps,twocolumn,nofootinbib]{revtex4-2}
\usepackage{mathtools}
\usepackage{multirow}
\usepackage{booktabs}
\usepackage{adjustbox} 
\usepackage{amssymb}
\usepackage{tikz}
\usepackage{amsthm}
\usepackage{amsmath}
\usepackage{array}
\usepackage{dcolumn}
\usepackage{bm}
\usepackage{graphicx}    
\usepackage{subcaption}  
\usepackage{enumitem}   
\usepackage{float}
\usepackage{makecell}  
\usepackage{braket} 
\usepackage{bm}
\usepackage{titlesec}
\titlespacing*{\subsubsection}{0pt}{0.6\baselineskip}{0.3\baselineskip}
\usepackage{titlesec}
\titlespacing*{\subsection}
  {0pt}   
  {0.8em} 
  {0.3em} 
\usepackage{hyperref}
\usepackage{url}
\usepackage{wrapfig}
\hypersetup{
    colorlinks,
    citecolor=blue,
    filecolor=blue,
    linkcolor=blue,
    urlcolor=blue,
}
\usepackage{ragged2e}
\renewcommand{\eqref}[1]{Eq.~({\ref{#1}})}

\DeclareMathOperator{\tr}{Tr}

\begin{document}

\author{Kaitlin Gili$^{1}$}
\email[Corresponding author: ]{kaitlinmgili@gmail.com}
\author{Zachary P. Bradshaw$^{1}$}
\affiliation{$^1$QodeX Quantum, Bethesda, MD, USA}

\date{\today}

\begin{abstract}
The field of quantum machine learning (QML) evolved to value models believed to most directly rival those providing utility in classical ML, namely large-scale neural networks. This pursuit of quantum neural-network analogues has produced important theoretical insights for expressivity and trainability at scale; however, due to our inability to conduct large-scale empirical studies, this direction has contributed to skepticism that QML will provide a compelling application for quantum computers. In the meantime, classical ML has been learning a hard lesson with respect to deploying un-interpretable neural networks in the wild: model interpretability matters for human adoption. We adopt this broader ML perspective to encourage a research direction in QML that prioritizes discovering and understanding \emph{tools} from quantum information theory that applied ML researchers can use to construct inherently interpretable models:~models with intentional mathematical structure that guarantees domain-aligned behavior. We first exemplify this research direction through a thread of work in a neighboring area---Bayesian ML. Specifically, we discuss an efficient mathematical \emph{tool} that emerged from a connection between Gaussian processes and Bayesian linear models, which has since been widely studied and used by applied researchers for problem-specific model design. We then discuss how prior QML research can be reframed as progress towards this direction, despite having been pursued with different goals in mind. Finally, we propose an initial sketch of future work in this direction, with the broader aim of better supporting the needs of applied ML researchers.
\end{abstract}

\title{Inherent interpretability provides inherent value in quantum machine learning}

\maketitle

\section{Introduction}

In the late 2010s, the field of quantum machine learning (QML) expanded in community and resources with a goal of delivering \emph{reasons} for the use of a near-term quantum computer \cite{Preskill2018, Bharti2022_nisqreview}. Given the rise of parameterized quantum circuits for optimization in the NISQ era \cite{Cerezo2021_vqareview, Moll2017, Wecker_PRA2015, lubasch2020variational}, it became natural to see these circuits as highly-flexible, trainable machine learning (ML) models that could potentially rival \emph{highly valued} classical neural networks for regression, classification, and generative tasks \cite{Farhi2018,cong2019quantumconvolutional, bausch2020recurrent, Benedetti2021, NonlinearQCBMs, chen2022quantumlong,  cherrat2024quantumvision}. What followed was a large effort to study the ML model quality of these circuits -- named quantum neural networks (QNNs) -- with different information-theoretic tools and generalization performance benchmarks \cite{abbas2021power,Gili2022,hur2022quantumconvolutional,khoo2024benchmarking, bowles2024better, basilewitsch2025quantum, qmetric2025benchmarking}. Through trainability investigations, we learned that uniformly random parameter initializations of circuits with high-order qubit interaction lead to a cost signal that requires a number of measurements that scales exponentially in the number of qubits to resolve \cite{schuld2019evaluating,Mcclear2018Barren, holmes2021scramblers, wang2021noiseinduced, holmes2022expressivity, cerezo2021costfunction, arrasmith2021gorges, larocca2025review}. In addition, we have learned that there are initialization methods that add structure to the problem, which in theory, should help with training large models \cite{shi2025avoiding}. Importantly, we learned that small-scale benchmarking studies do not provide much evidence that QNNs are going to be useful at scale, or offer knowledge of why and when they should be used \cite{bowles2024better}. Furthermore, the quantum community has adopted considerable skepticism that QNNs will provide us with a reason to use a quantum computer: small scale $\rightarrow$ classically simulatable, large-scale $\rightarrow$ we won't know until we can try. 

In the meantime, the classical ML community has not concerned itself with rivaling quantum, but instead with advancing, adapting, and deploying deep learning models in the wild \cite{Bengio-Book, bengio2021deep, paleyes2022challenges}. In doing so, the classical community has run into a whole new set of problems that come from developing a technology that is intended to enhance \emph{human} well-being and productivity. For example, in high-risk contexts, we want the model to learn to abstain from making a prediction if it doesn't have enough information in its training data to do so \cite{hendrickx2024machine, hasan2023survey}. This is especially relevant in medical prognosis, where we largely lack training data for certain demographics (e.g.~gender, race), and yet, we want to build models that benefit everyone \cite{norori2021addressing, turner2022race}. Alongside abstention, there has been a large push to understand \emph{how} the model is making its predictions in order to establish trust with users. This has led to the rapidly expanding field of mechanistic interpretability, where research focuses on building post-hoc explanations of large-scale neural networks \cite{bereska2024mechanistic,somvanshi2026bridging}. 

More recently, there has been a call to shift away from developing post-hoc explanations of uninterpretable models and to move towards \emph{inherently} interpretable ones \cite{Zschech2025Inherently}. An inherently interpretable ML model is ``constrained in model form so that it is either useful to someone, or obeys structural knowledge of the domain" \cite{Rudin2019StopExplaining}. Put another way, model structure is intentionally shaped in order to produce a soft or hard bias towards problem-specific behavior. This approach benefits both applied ML researchers and practitioners, as well as users. Researchers and practitioners can sit at the intersection of theory and experiment -- principled mathematical theory informs model building for their problem-specific context, and experiments can be used to validate the theory on real-world data, and iterate for further improvement. Users can place greater confidence in a tool that is built with domain expertise, rather than a black-box that no one understands. 

We adopt this larger ML perspective to encourage a research direction in QML which prioritizes discovering and understanding \emph{tools} from quantum information theory that applied ML researchers can use to construct inherently interpretable models. In this frame of thinking, one is less concerned with the development of large-scale QNNs to rival classical ones in generalization performance, or the development of any new model that is to be selected by an applied researcher or practitioner to simply work ``off the shelf". Rather, one prioritizes the question:~\emph{what tools from quantum information enable the manipulation of model structure to support problem-specific model design in an ML setting?} Then, provided such tools and an understanding of the classical or quantum computational resources required to use them, applied ML researchers and practitioners can design models that are well-suited to their own specific problems. 

We first exemplify this research direction through a thread of work in Bayesian ML, in which an efficient \emph{tool} emerged from a connection between Gaussian processes and Bayesian linear models, and continues to be studied and used by applied researchers for problem-specific model design. We then discuss how prior QML research can be reframed as progress towards this direction, despite having been pursued with different goals in mind. Finally, we propose an initial sketch of future work in this direction. Overall, we aim to encourage a direction in which the mathematical tools from quantum information and the capabilities of quantum hardware better support the needs of applied ML researchers and practitioners. 

\section{An example research thread from Bayesian ML}

We exemplify our proposed research direction through a thread of work in Bayesian ML, in which an efficient mathematical \emph{tool} emerged from an established theoretical connection, and continues to be studied and used for problem-specific model design. We first present this connection---between Gaussian processes (GP) and Bayesian linear models--- and explicitly define the tool, before discussing its applied use. 

\subsection{The connection: GPs through Bayesian linear models}\label{Bayesian_GPs}

Let us assume that we have access to a regression training dataset $\mathcal{D} = \{x_i, y_i\}_{i=1}^N$, where $x_i \in \mathbb{R}^d$ and $y_i \in \mathbb{R}$. We transform the input data via a chosen non-linear feature map $x_i \mapsto \phi(x_i) \in \mathbb{R}^R$. The relationship in the data is modeled via an inner product between the features and an $R-$dimensional weight vector $c$, represented as: 
\begin{equation}\label{base_model}
    f_{c}(x_i)  = c^T\phi(x_i) = \sum_{r=1}^R c_r \phi_r(x_i).
\end{equation}
In a typical supervised linear regression setting, we train the parameters $c$ to minimize a loss function of the form $\mathcal{L}(c) = \sum_{i=1}^N\ell(y_i, f_c(x_i))$, where $\ell$ is typically the mean squared error (MSE) function.

In the Bayesian approach to this problem, we keep the same base model $f_{c}(x_i)$, but introduce a probabilistic framework to maximize the likelihood of the model generating each true data value $y_i$, under some uncertainty over the model weights. We represent our Bayesian framework by defining the following distributions: 
\begin{equation}
\begin{array}{@{}l@{\;}l@{}}
\text{Prior:}
\ p(c) =  \mathcal{N}(c| 0_R,I_{R \times R}), \\[0.35em]
\text{Joint likelihood:} \ p(y|c) = 
 \prod_{i=1}^N\mathcal{N}(y_i| c^T\phi(x_i), \sigma_\text{noise}^2). \\[0.35em]
\end{array}
\label{distributions}
\end{equation}
\normalsize
We define the prior over parameters $c$ to be multivariate Gaussian with an $R-$dimensional zero mean vector and an $R \times R$ identity covariance matrix, as indicated by each subscript. The identity covariance matrix incorporates the assumption that the parameters of our model $c_r$ are independent. The joint likelihood is constructed as a product of Gaussian distributions for each data input, where the mean is the scalar model output $f_c(x_i)$ with a shared scalar hyper-parameter variance $\sigma^2_{\text{noise}}$. Thus, we bake in the assumption that $y_i = f_c(x_i) + \epsilon_i$, where the scalar $\epsilon_i \sim \mathcal{N}(0, \sigma_\text{noise}^2)$. One can show that this assumption exists inherently when minimizing the MSE loss ($\epsilon_i$ is the minimized error) in regression problems, and that maximizing the $\log$ Gaussian likelihood function is equivalent to minimizing MSE up to an additive constant. Note that we condition on the input data $x$, rather than treat it as a random variable, so we do not include it in our probabilistic model definitions. For those unfamiliar with a Bayesian approach to ML, a soft introduction is supplied in Appendix~\ref{soft_intro_Bayesian}.

According to Bayes rule, the product of the joint likelihood and the prior are proportional to the Gaussian posterior. When both distributions are Gaussian, the posterior mean $\mu_{\text{post}}$ and covariance $\Sigma_{\text{post}}$ have the following closed form: 
\begin{equation}
\begin{array}{@{}l@{\;}l@{}}
\text{Posterior $p(c|y) = \mathcal{N}(c|\mu_{\text{post}},\Sigma_{\text{post}})$:} \\
\ \text{$\mu_{\text{post}}$} =  \frac{1}{\sigma_{\text{noise}}^2}\Sigma_{\text{post}}\phi(x)^T y, \\[0.35em]
\text{$\Sigma_{\text{post}}$} = (I + \frac{1}{\sigma_{\text{noise}}^2}\phi(x)^T\phi(x))^{-1}.  \\[0.35em]
\end{array}
\label{posterior}
\end{equation}
The posterior distribution captures uncertainty in the model parameters $c$ after observing all training inputs $x$ and outputs $y$. Here, $\phi(x)$ is a $N \times R$ matrix, containing the feature vectors of the entire dataset. To predict $f_c(x_i^{\text{test}})$ on a new single test point $x_i^{\text{test}}$, we transform the data into $\phi(x_i^{\text{test}})$ and compute the closed form Gaussian mean $\mu_{\text{pred}}$ and variance $\sigma^2_{\text{pred}}$ for the posterior predictive distribution: 
\begin{equation}
\label{posterior_predictive}
\begin{aligned}
p\!\left(f_c(x_i^{\text{test}}) \mid y\right)
&=
\int p\!\left(f_c(x_i^{\text{test}}) \mid c\right)
p(c \mid y)\,dc \\
&=
\mathcal{N}\!\left(
f_c(x_i^{\text{test}}) \mid \mu_{\mathrm{pred}}, \sigma^2_{\mathrm{pred}} 
\right),
\\[0.5em]
\mu_{\mathrm{pred}}
&=
\phi(x_i^{\text{test}})^T \mu_{\mathrm{post}},
\\[0.35em]
\sigma^2_{\mathrm{pred}}
&=
\phi(x_i^{\text{test}})^T
\Sigma_{\mathrm{post}}
\phi(x_i^{\text{test}}).
\end{aligned}
\normalsize
\end{equation}

The variance of the model prediction for the new test point depends on its feature embedding directional overlap with that of the training data. The uncertainty decreases in the directions defined by the feature vectors of the training data. If the test embedding is equal to a training embedding, the uncertainty over the new test point is low. If the test embedding is completely orthogonal to all training embeddings, the uncertainty over the new test point is $\phi(x_i^{\text{test}})^T\phi(x_i^{\text{test}})$. Thus, the uncertainty is calibrated based on how well the feature directions of the test data overlap with that of the training data. 

So far, we have a described a typical Bayesian linear model for a supervised regression problem \cite{bishop2006pattern}. With little effort, we can shift our \emph{interpretation} of the mathematics. Consider that our independent weights $c_r$ from Eq.~\ref{base_model} are values sampled from a Gaussian. We observe that $f_c(x_i)$ is a linear sum over scalars multiplied by Gaussian distributed values. By the closure of Gaussian distributions under affine transformations, we know that the function output $f_c(x_i)$ is Gaussian distributed. More specifically, for a single datapoint $x_i$ we know by our prior distribution assumptions that $\mathbb{E}_{c_r \sim p(c_r)}[c_r] = 0$ and $\mathbb{E}_{c_r \sim p(c_r)}[c_r^2] = 1$. As a result, we obtain $\mathbb{E}_{c_r \sim p(c_r)}[\phi_r(x_i)c_r] = 0$ and $Var[\phi_r(x_i)c_r] = \phi_r(x_i)^2$. The variance comes from the notion that $\phi_r(x_i)$ is not a random variable, so it can be pulled out of the expectation to obtain $\phi_r(x_i)^2\mathbb{E}[c_r^2]$. If we use the same logic as above, over \emph{all data} $x$, such that $f_c(x) = \phi(x)c$ where $\phi(x)$ is a $N \times R$ matrix and $f_c(x)$ is output vector of dimension $N$, we observe that $f_c(x) \sim \mathcal{N}(0_N, \phi(x)\phi(x)^T)$. Each entry of the prior covariance matrix is an inner product over $R$ features, and thus a finite-dimensional kernel function $k_{\phi}(x_i, x_i')$. Thus, defining a Gaussian prior distribution over parameters $c \sim p(c)$ induces a finite GP prior over functions $f_c(x) \sim \mathcal GP(0_N, k_{\phi}(x_i, x_i'))$. 

A GP is generally defined as a Gaussian distribution over functions, $f(x) \sim \mathcal GP(\mu, k(x_i, x_i'))$, where $\mu$ is a vector of mean function evaluations and the covariance matrix is determined by pairwise input similarities encoded by the kernel function $k(x_i,x_i')$. GPs have long been the gold-standard for meaningful uncertainty quantification over predictions on unseen data, as they incorporate relationships between the input data and the uncertainty estimates of the output data predictions \cite{rasmussen2006Gaussian}. Hence, the model directly encodes the bias that new inputs that are dissimilar to the training inputs will have high uncertainty over their predictions. The chosen kernel encodes what dissimilarity means with respect to the data. 

GPs are computationally expensive to compute directly via a chosen kernel beyond a few thousand datapoints, which a GP induced by a Bayesian linear model can address by requiring inversion of a matrix sized by the number of features rather than datapoints \cite{rasmussen2006Gaussian}. The conceptual shift from weight-space to function-space does not change the numerical result of the posterior predictive in Eq.~\ref{posterior_predictive} \footnote{The posterior predictive can be re-written mathematically in function space, where the inverted matrix will be over a $N \times N$ kernel matrix rather than the $R \times R$ feature matrix. The closed form mean and covariance expressions are numerically equivalent. See Appendix \ref{exact_inference} for the function-space frame.}, but rather it provides an alternative way to understand the implications of choosing the feature map $x_i \mapsto \phi(x_i)$ in the efficient computation of the GP. More specifically, the feature map controls the induced GP prior covariance, encoding how training and test datapoints are related to one another into the GP uncertainty estimates. In the next section, we show how choices in the structure of the feature map, i.e.~the aforementioned \emph{tool}, directly influence GP design. 

\subsection{The tool: manipulating Fourier feature structure}\label{characterization}

In this section, we explore the structure of random Fourier feature (RFF) maps \cite{rahimi2007random}, commonly used in a Bayesian-GP induced setting, and show how the choice of frequency distribution serves as an \emph{tool} for manipulating this structure to shape the resulting kernel. As put forth in \cite{rahimi2007random}, RFF models consist of finite feature embeddings $\phi^{\small RFF}(x_i)$ that can be written as: 
\begin{equation}
\phi^{\small RFF}(x_i) = 
\begin{bmatrix}
2\cos(x_i\omega_1^T)\\
2\sin(x_i\omega_1^T)\\
\vdots \\
2\cos(x_i\omega_{|\Omega|}^T)\\
2\sin(x_i\omega_{|\Omega|}^T)
\end{bmatrix}.
\label{RFF_feature_map}
\end{equation}

The feature map transforms input data $x_i \in \mathbb{R}^d$ to a vector of dimension $R=2|\Omega|$, where $|\Omega|$ is the size of the set of frequency vectors $\{\omega_z\}_{z=1}^{|\Omega|}$. Each frequency vector $\omega_z \in \mathbb{R}^d$.

When substituted into Eq.~\ref{base_model}, each feature coefficient $c_r$ is separated into $c^{\alpha}_z$ and $c^{\beta}_z$, which are independent feature coefficients for the $\cos$ and $\sin$ terms, per frequency $z$. In accordance with our previous construction, $c^{\alpha}_z,c^{\beta}_z \sim \mathcal{N}(0,1)$. Upon substituting these components into Eq.~\ref{base_model}, we observe that the model is written as:  
\begin{equation}
f_{c}(x_i) = \sum_{z=1}^{|\Omega|} 2c_z^{\alpha} \cos(\omega_z^Tx_i) + 2c_z^{\beta}\sin{(\omega_z^Tx_i)}.
\label{fourier_model}
\end{equation}
An equivalent form of this RFF model is the Fourier-type model, written as: 
\begin{equation}
f_{g}(x_i) = \sum_{z=1}^{|\Omega|} g_z e^{ix_i\omega_z^T} + g^*_z e^{-ix_i\omega_z^T}
\label{fourier_model_2}.
\end{equation}
The model can be viewed as a sum of Fourier basis functions, where each is defined by a distinct frequency vector $\omega_z$. This form reveals the \emph{main} assumption baked into our RFF model: besides the case when $\omega_z, -\omega_z = 0$, we \emph{always} have frequency pairs $\omega_z, -\omega_z$ with Fourier coefficients that are complex conjugate pairs. The complex Fourier coefficients $g_z$ are related to our real coefficients as: $g_z = c^{\alpha}_{z} - ic^{\beta}_{z}$. 

Defining the kernel function $k_{\phi}(x_i,x'_i)$ via this feature map inner product results in: 
\begin{equation}
\begin{aligned}
\frac{1}{R}\phi^{\mathrm{RFF}}(x_i)\phi^{\mathrm{RFF}}(x_i')^T
&= \frac{1}{|\Omega|}\sum_{z=1}^{|\Omega|}
2\cos(\omega_z^T x_i)\cos(\omega_z^T x_i')
\\
&\quad + \frac{1}{|\Omega|}\sum_{z=1}^{|\Omega|}
2\sin(\omega_z^T x_i)\sin(\omega_z^T x_i')
\\
&= \frac{1}{|\Omega|}\sum_{z=1}^{|\Omega|}
2\cos\!\left(\omega_z^T(x_i-x_i')\right).
\end{aligned}
\label{shift_in_kernel}
\end{equation}
The kernel is normalized by the number of features. Since relative proportion information with respect to the density $p(\omega_z)$ is incorporated into the sum via the number of terms that correspond to each frequency sample $\omega_z$ and we have the normalization term $\frac{1}{|\Omega|}$, we observe that Eq.~\ref{shift_in_kernel} is an unbiased Monte-Carlo approximation of the expectation $\mathbb{E}_{\omega_z \sim p(\omega_z)}[2\cos(\omega_z^T(x_i-x_i'))]$. It converges to the expectation as $|\Omega| \rightarrow \infty$.  Importantly, when the frequency density is symmetric $p(\omega_z) = p(-\omega_z)$, we have the following equivalence : 
\begin{equation}
\begin{aligned}
k_{\text{target}}(x_i, x_i')
&= \mathbb{E}_{\omega_z \sim p(\omega_z)}
\left[
2\cos\!\left(\omega_z^T(x_i-x_i')\right)
\right]
\\
&= 2\int p(\omega_z)
e^{i\omega_z^T(x_i-x_i')}
\, d\omega_z .
\end{aligned}
\label{fourier_transform}
\end{equation}

This follows \emph{Bochner's theorem}, which states that any continuous positive-definite shift-invariant kernel function can be written as the Fourier transform of a nonnegative spectral measure. We observe that the kernel function $k_{\phi}(x_i, x_i')$ that controls the GP prior covariance can be controlled via a chosen spectral distribution $p(\omega_z)$. In RFF, one chooses this distribution and samples $|\Omega|$ frequencies from it to approximate a target kernel $k_{\text{target}}(x_i, x_i')$ that encodes a desired notion of data similarity informed by problem-specific knowledge. Below, we discuss how researchers have studied and used this tool for problem-specific model design in uncertainty quantification. 

\subsection{Studying and using the tool for problem-specific model design}\label{characterization}

Study and use of spectral distribution choices for GP design dates back to the early 2000s.~For example, MacKay et al.~\cite{mackay2003information} showed that when $p(\omega_z)$ is a Gaussian distribution, the integral becomes a Radial Basis Function (RBF) kernel. This kernel contains a bias of \emph{distance awareness} between the training and test data inputs. More specifically, the GP is biased to have higher variance for test points that are farther away in data-space from the training data. Liu et al.~\cite{liu2023simple} proposed that this bias be used in the last layer of a neural network to obtain distance-aware uncertainty estimates over black-box predictions. Zanetta et al.~\cite{zanetta2026efficient} used a Gaussian spectral distribution to induce an RBF kernel alongside a deep learning kernel to model wind gusts, where the RBF kernel is meant to add a bias based on spatial locality for the wind gust observations. Harvey et al.~\cite{harvey2026learning} showed that when the number of features is much larger than the data size, one can use a data-emphasized variational inference (VI) method to implicitly learn the hyper-parameters of this RBF kernel. Wilson et al.~\cite{wilson2013Gaussian} explored opportunities for kernel discovery when $p(\omega_z)$ is a Gaussian mixture and the hyper-parameters are learned via Bayesian VI methods. In \cite{lazaro-gredilla2010sparse}, the authors explored an initialization method for the frequencies that requires sampling them from a chosen distribution that directly relates to the target kernel, and \emph{then} learning them directly as hyper-parameters with VI. Gal et al.~\cite{gal2015improving} improves upon this method by placing a variational posterior over the frequencies, so that their uncertainty is directly incorporated into the posterior. Borovitsky et al.~\cite{borovitskiy2020matern} utilized a discrete Gaussian with integer frequencies to induce a \emph{periodic} bias into the kernel of uncertainty estimates over pendulum Hamiltonian dynamics. Chau et al.~\cite{chau2021deconditional} utilized a Student-$t$ distribution to induce a Matérn kernel, which incorporates desired \emph{roughness} into the co-variates of atmospheric observations.  

The research thread above illustrates how an efficient tool can emerge from intersecting mathematical theories, which can then be further studied and used for problem-specific model design. We believe analogous research can be conducted in QML using tools from quantum information theory, by which we mean broadly to include tools from spectral analysis \cite{belis2026spectral}, group theory \cite{Larocca_2022}, representation theory \cite{ragone2023representation}, kernel theory \cite{schuld2021supervised}, and topology \cite{gyurik2022towards} \emph{applied} to quantum information. Before discussing future work, we show how prior QML research drawing on these mathematical theories can be reframed as progress towards this direction.

\section{Prior QML work as progress in this direction}

In this section, we discuss how prior work in QML can be reframed to support the research direction we aim to encourage. We focus on two research threads:~geometric QML and quantum feature maps. These are by no means the only threads in QML that can support this direction; they simply offer two particularly clear examples.

\subsection{Geometric QML}

In standard QNNs, or more generally in variational QML, a generic parameterized circuit $U(\theta)$ realizes an essentially unstructured map from data to expectation values; whereas geometric quantum machine learning (GQML)~\cite{meyer2022exploiting,Larocca_2022,ragone2023representation, nguyen2024theory} demands that the model respect a symmetry group $G$ acting on the data. For the uninitiated, we provide a soft introduction to representation theory in Appendix~\ref{app:rep-theory}.

Let $G$ be a group acting on some set of data $\mathcal{X}$. This action and a choice of encoding $\mathcal{E}:\mathcal{X}\to\mathcal{H}$ together induce a unitary representation $R:G\to\mathcal{U}(\mathcal{H})$ defined by
\begin{equation}\label{eq:induced-representation}
    \mathcal{E}(g\cdot x_i)=R(g)\mathcal{E}(x_i) 
\end{equation}
so that the representation $R$ is thought to carry the action of $G$ on the level of the encoding Hilbert space. This transfer of structure allows us to understand the symmetry properties of a data point $x_i \in \mathbb{R}^d$ from a finite dataset after it has been encoded into a quantum state. Indeed, if the datapoint $x_i$ is left invariant by $g$, then by \eqref{eq:induced-representation}, we also have $R(g)\mathcal{E}(x_i)=\mathcal{E}(x_i)$. If $\mathcal{E}$ is injective, then the converse also holds and the two notions of invariance are in fact equivalent. 

In GQML, this transfer of the group action structure is used to construct \emph{equivariant} QNNs. To illustrate this point, let us suppose our model has the typical form~\cite{Cerezo2021_vqareview}
\begin{equation}
    f(x_i)=\tr[U(\theta)\mathcal{E}(x_i)\mathcal{E}(x_i)^\dagger U^\dagger(\theta)\mathcal{O}],
\end{equation}
where $U(\theta)$ is some choice of variational ansatz and $\mathcal{O}$ is a measurement observable. The variational ansatz is called equivariant when it commutes with the representation. That is,
\begin{equation}\label{eq:equivariant-commutator}
    [U(\theta),R(g)]=0
\end{equation}
for all possible parameters $\theta$ and all $g\in G$. We assume $U(\theta)$ is a product of unitary transformations $U_k(\theta_k)$, so that each such transformation can be written $U_k(\theta_k)=e^{-i\theta_k H_k}$ for some Hermitian operator $H_k$ by Stone's theorem~\cite{stone1932one}. By Taylor expanding, it follows that a sufficient condition for \eqref{eq:equivariant-commutator} is
\begin{equation}
    [H_k,R(g)]=0
\end{equation}
for all $g\in G$, independent of $\theta$. If in addition to an equivariant ansatz, one has an equivariant observable ($[\mathcal{O},R(g)]=0$ for all $g\in G$), then the model output is invariant under the action of the symmetry group. Indeed, by the cyclicity of the trace and the equivariance conditions, we have
\begin{align}
    f(g\cdot x_i)&=\tr[U(\theta)\mathcal{E}(g\cdot x_i)\mathcal{E}(g\cdot x_i)^\dagger U^\dagger(\theta)\mathcal{O}]\\
    &=\tr[U(\theta)R(g)\mathcal{E}(x_i)\mathcal{E}(x_i)^\dagger R^\dagger(g)U^\dagger(\theta)\mathcal{O}]\\
    &=\tr[U(\theta)\mathcal{E}(x_i)\mathcal{E}(x_i)^\dagger U^\dagger (\theta) R^\dagger(g)R(g)\mathcal{O}]\\
    &=\tr[U(\theta)\mathcal{E}(x_i)\mathcal{E}(x_i)^\dagger U^\dagger (\theta)\mathcal{O}]=f(x_i).
\end{align}

The response of the model to any symmetry transformation of the input is fixed by construction rather than probed empirically. This is further evidenced in Rodriguez-Grasa et al.~\cite{rodriguezgrasa2026pacbayesianapproachgeneralizationquantum}, which recently derived PAC bounds that characterize equivariant quantum models as having a restricted hypothesis class. 

Beyond constructing equivariant models with representation-theoretic tools, choices of operator properties and structure, without an explicitly specified underlying group symmetry, can also serve as tools for guaranteeing behavior. For example, Bowles et al.~\cite{bowles2023contextuality} encoded a label-level conservation law into the measured observable $\mathcal{O}$, and constrained the remaining circuit architecture to commute with that observable so that the problem-specific conserved quantity was preserved throughout the computation. Additionally, Gili et al.~\cite{gili2024inductive} studied the inductive bias induced by non-commuting measurement observables $\mathcal{O}$ in a generative model for order effects in binary question--answer data. The work suggests that commutation properties of these observables can be used to encode whether the model's predictions depend on the order of the questions. Overall, the key point is that representation-theoretic and operator-level structure can be manipulated to guarantee model behavior, meaning that we already have tools for problem-specific design in QML to further explore and use.

\subsection{Quantum feature maps}
 A quantum model never acts directly on raw classical data. Even when it is designed to preserve an identity feature map on the classical inputs \cite{wang2026explainablequantumregressionalgorithm}, it ultimately acts on the \emph{image} of an encoding map, $\mathcal{E}:\mathcal{X}\to\mathcal{H}$.
 Let us allow for mixed states so that the feature map instead takes the form $\rho:\mathcal{X}\to S(\mathcal{H})$, where $S(\mathcal{H})$ is the set of all density matrices. We further assume that the data space is a metric space with metric $d_\mathcal{X}$~\cite{searcoid2007metric,rudin2021principles,kaplansky2020set}). After the feature map, whatever the subsequent trainable circuit does, the interpretive question of how nearby inputs are treated is settled by the metric behavior of this feature map. This is the central concern of robustness and stability, for example. Requiring that $\rho$ distort distances within fixed bounds turns that behavior from an emergent accident into an inherently designed guarantee.

Explicitly, we equip the space of density matrices with the trace distance $d_{\mathcal{S}(\mathcal{H})}(\rho,\sigma)=\frac12\|\rho-\sigma\|_1$ induced by the Schatten 1-norm~\cite{watrous2018theory}. This distance function quantifies the distinguishability of quantum states. Controlling the distortion of distances is captured by the bi-Lipschitz condition
\begin{equation}
    c_- d_\mathcal{X}(x_i,x_i')\le d_{\mathcal{S}(\mathcal{H})}(\rho(x_i),\rho(x_i'))\le c_+d_{\mathcal{X}}(x_i,x_i')
\end{equation}
for some constants $0<c_-\le c_+$. The two inequalities carry distinct interpretive content. The upper bound makes $\rho$ Lipschitz; nearby inputs map to indistinguishable states and the model cannot react violently to small perturbations. Meanwhile, the lower bound is a faithfulness guarantee; distinct inputs remain quantitatively distinguishable and the embedding never silently collapses information the task may need.

Observe that a subsequent variational unitary does not change either bound for the embedding of datapoint $x_i \in \mathcal{X}$. Indeed, we have
\begin{equation}
\begin{aligned}
    d_{\mathcal{S}(\mathcal{H})}&(U(\theta)\rho(x_i)U^\dagger(\theta),U(\theta)\rho(x_i')U^\dagger(\theta))\\
    &=\frac12\|U(\theta)\rho(x_i)U^\dagger(\theta)-U(\theta)\rho(x_i')U^\dagger(\theta)\|_1\\
    &=\frac12\|U(\theta)(\rho(x_i)-\rho(x_i'))U^\dagger(\theta)\|_1\\
    &=\frac12\|\rho(x_i)-\rho(x_i')\|_1=d_{\mathcal{S}(\mathcal{H})}(\rho(x_i),\rho(x_i')),
\end{aligned}
\end{equation}
so the trainable parameters only rotate the embedded data geometry and never stretch it. It follows that the full model
\begin{equation}
    f(x_i)=\tr[\mathcal{O}U(\theta)\rho(x_i)U^\dagger(\theta)]
\end{equation}
inherits an explicit, parameter independent notion of continuity. Explicitly, we have
\begin{equation}
\begin{aligned}
    |f(x_i)-f(x_i')|&=|\tr[\mathcal{O}U(\theta)(\rho(x_i)-\rho(x_i'))U^\dagger(\theta)]|\\
    &\le\|\mathcal{O}\|_\infty\|U(\theta)(\rho(x_i)-\rho(x_i'))U^\dagger(\theta)\|_1\\
    &=\|\mathcal{O}\|_\infty\|\rho(x_i)-\rho(x_i')\|_1\\
    &\le2c_+\|\mathcal{O}\|_\infty d_\mathcal{X}(x_i,x_i'),
\end{aligned}
\end{equation}
where we have applied Holder's inequality~\cite{rudin1987real} and the Lipschitz condition. Thus, $d_\mathcal{X}(x_i,x_i')<\epsilon (2c_+\|\mathcal{O}\|_\infty)^{-1}$ implies $|f(x_i)-f(x_i')|<\epsilon$. The constant $2c_+\|\mathcal{O}\|_\infty$ is a global interpretability statement that yields certified robustness radii and smoothness readable off the encoding and the observable norm without probing the trained weights. 

The robustness statement is particularly valuable and deserves closer inspection. Suppose that the model predicts a classification of $x_i$ by thresholding the score, so that the class prediction is given by $\textnormal{sign}(f(x_i)-t)$. Define the margin by $m(x_i)=|f(x_i)-t|$, the distance of the score from the decision threshold. To change the decision, the model must move $x_i$ to an $x_i'$ beyond the margin so that $|f(x_i')-f(x_i)|\ge m(x_i)$. But the Lipschitz condition caps the change at $2c_+\|\mathcal{O}\|_\infty d(x_i,x_i')$, so a flip requires
\begin{equation}
    d(x_i,x_i')\ge\frac{m(x_i)}{2c_+\|\mathcal{O}\|_\infty}.
\end{equation}
Thus, every input within the ball of radius 
\begin{equation}
r(x_i)=\frac{|f(x_i)-t|}{2c_+\|\mathcal{O}\|_\infty}
\end{equation}
is guaranteed the same prediction. A tighter upper-distortion constant $c_+$ (or a smaller observable norm) enlarges the certified ball, which is the concrete reason controlling $c_+$ via the data encoding buys robustness.

As a worked example, take $\mathcal{X}=[0,\pi]$ with the Euclidean metric $d_\mathcal{X}(x,x')=|x-x'|$ and the angle encoding $\ket{\phi(x)}=e^{-i xY/2}\ket{0}=\cos\!\left(\tfrac{x}{2}\right)\ket{0}+\sin\!\left(\tfrac{x}{2}\right)\ket{1}$, with $\rho(x)=\ket{\phi(x)}\!\!\bra{\phi(x)}$. For pure states, the trace distance is $d_{\mathcal{S}(\mathcal{H})}(\rho(x),\rho(x'))=\sqrt{1-|\braket{\phi(x)|\phi(x')}|^2}$, and since $\braket{\phi(x)|\phi(x')}=\cos\!\left(\tfrac{x-x'}{2}\right)$, we obtain the exact expression
\begin{equation}
    d_{\mathcal{S}(\mathcal{H})}(\rho(x),\rho(x'))=\left|\sin\!\left(\tfrac{x-x'}{2}\right)\right|.
\end{equation}
The argument $\tfrac{|x-x'|}{2}$ ranges over $[0,\tfrac{\pi}{2}]$, where the bounds $\tfrac{2u}{\pi}\le\sin u\le u$ hold, and hence
\begin{equation}
    \tfrac{1}{\pi}\,|x-x'|\;\le\; d_{\mathcal{S}(\mathcal{H})}(\rho(x),\rho(x'))\;\le\;\tfrac{1}{2}\,|x-x'|,
\end{equation}
so the encoding is bi-Lipschitz with $c_-=\tfrac{1}{\pi}$ and $c_+=\tfrac{1}{2}$. The two interpretive guarantees are now explicit numbers: no pair of inputs is mapped closer together than $\tfrac{1}{\pi}$ times their data-space separation (faithfulness), and no perturbation is amplified by more than a factor of $\tfrac{1}{2}$ (stability). For the observable $\mathcal{O}=Z$ with $\|\mathcal{O}\|_\infty=1$, the model satisfies $|f(x)-f(x')|\le 2c_+\|\mathcal{O}\|_\infty|x-x'|=|x-x'|$ for every parameter value $\theta$, and a thresholded classifier with margin $m(x)=|f(x)-t|$ carries the certified radius
\begin{equation}
    r(x)=\frac{m(x)}{2c_+\|\mathcal{O}\|_\infty}=m(x),
\end{equation}
readable before training and unchanged by it. 

From the kernel perspective~\cite{schuld2021supervised}, the Lipschitz metric condition is a statement about the kernel function induced by the embedding. For a pure state encoding $x_i\mapsto\ket{\phi(x_i)}$, the fidelity kernel
\begin{equation}
    k(x_i,x_i')=|\braket{\phi(x_i)\vert\phi(x_i')}|^2
\end{equation}
measures how similar two encoded points are, and every state-space distance is a function of it. For instance, 
\begin{equation}
d_{\mathcal{S}(\mathcal{H})}(\ket{\phi(x_i)}\!\!\bra{\phi(x_i)},\ket{\phi(x_i')}\!\!\bra{\phi(x_i')})=\sqrt{1-k(x_i,x_i')}.
\end{equation}
Bounded distortion simply says that this similarity tracks the closeness in the data space; nearby inputs give $k$ near 1, distant inputs give $k$ near 0, with the \emph{bounded global rate} determined by the distortion constants. Put simply, the constants can be viewed as a tool for problem-specific design, bounding how much the embedding $\rho(x_i)$ changes as the data $x_i$ changes, and thus the speed at which fidelity kernel values change as a function of distance in the data. 

\subsection{Initial sketch of future work in this direction}

We encourage a long-term research program in QML focused on discovering and understanding tools for problem-specific model design that enable applied ML researchers to construct inherently interpretable models. As an end goal, we envision a \emph{toolkit} that describes how model structure can be manipulated to encode problem-specific knowledge, what behavior these choices guarantee, and what computational resources are required to implement the resulting design. Applied researchers and practitioners could use these tools to build models around the knowledge available to them, while accounting for computational efficiency.

For QML, this program encompasses mathematical tools from quantum information theory and its underlying mathematical theories, alongside the quantum hardware required to implement particular designs. This raises an intriguing question:~\emph{could quantum hardware be necessary not to learn a distribution, but to design the family of distributions within which learning takes place?}

As a more concrete next step, Belis et al.~\cite{belis2026spectral} propose an interesting idea that could be extended in this direction: tools from spectral theory can induce a simplicity bias for ML regularization and may be natural to implement on quantum computers. For example, the Fourier transform over the Boolean cube maps a data distribution into Fourier space, where filtering can suppress high-frequency components prior to efficient sampling. But what if low-pass filtering were instead used as a tool for model design of an initial distribution family, before learning the model parameters? One might have problem-specific knowledge about the Laplacian geometry of bitstrings---that is, which bitstring changes should count as small deviations \cite{bradshaw2026smoothness}. This knowledge could define a Laplacian whose eigenvalues identify which Fourier components to suppress, thereby restricting the model to distributions with the desired smoothness before learning begins.

This raises both modeling and computational questions:~\emph{If one is to smooth a chosen distribution family prior to learning, when is it more computationally efficient to go into Fourier space to smooth rather than applying a heat diffusion kernel?}~\emph{When the model distribution is first parametrized by a quantum circuit antaze, how does the smoothing procedure on the quantum state influence the final learned distribution in a meaningful way, and when is the smoothing procedure classically simulatable?} Notice that these questions are raised in the setting in which the smooth model is \emph{not} constructed directly in Fourier space, where the Fourier coefficients themselves require initial definition before learning. In the quantum setting, if the number of Fourier coefficients is $poly(n)$, where $n$ is the number of qubits in the circuit, then learning and sampling in Fourier space is classically simulatable. Thus, we can also ask the question of \emph{how can models constructed in Fourier spaces associated with different group structures encode problem-specific knowledge while retaining a sufficiently rich model family to potentially warrant quantum computation?}
We view an answer to any of these questions as advances for both machine learning and quantum computation.

In addition to the thread above, we are interested in understanding how Hamiltonian choices in quantum Fourier models \cite{schuld2021effect, oh2026fourier, strobl2025fourier}---the quantum analogue of RFF---serve as a tool for designing frequency structure, and how other circuit choices can shape the Fourier representation to guarantee desired model behavior.

\section{Conclusion} 

We hope this perspective offers readers a different lens through which to understand the \emph{value} that quantum information can bring to machine learning. Rather than constructing uninterpretable QNNs and developing post-hoc methods to explain them, we encourage understanding the tools within QML that enable researchers to design models with guarantees on behavior relevant to their problems. Rather than expecting applied researchers and practitioners to select models off the shelf and run them on their data, we should recognize the domain knowledge they bring and provide tools they can adapt to design models for their own contexts.

We hope the Bayesian ML example demonstrates that this kind of research is both possible and useful, and that researchers already working in QML can connect their contributions to this direction. Finally, we hope the initial sketch of future work inspires others in the field. We recognize that this research program is still in its infancy and that ideas from across the ML and quantum communities will be essential to realizing its full potential.

\begin{acknowledgments}
We wish to acknowledge Dr.~Ethan Harvey for initial manuscript feedback, and specifically pointing us to related works that propose RFF approaches and Bayesian techniques for Gaussian process design. We are grateful to ChatGPT for its feedback on the manuscript, which helped us with notation clarity and consistency, as well as the articulation of our ideas.
 \end{acknowledgments}





\appendix

\section{A soft introduction to Bayesian methods in ML}\label{soft_intro_Bayesian}

\subsection{Probabilistic machine learning}

Bayesian methods offer pathways to bake structure into ML settings, but they require one to view a standard regression or classification problem as one that learns parameters of a probability distribution or density function. Consider a training dataset $\mathcal{D} = \{x_i, y_i\}_{i=1}^N$, where $x_i \in \mathbb{R}^d$ and $y_i \in \{l_1,\dots,l_K\}$. Here, $y_i$ is one of $K$ categories. We choose to model the relationship between the inputs and outputs by a simple linear model: 
\begin{equation}
    \hat{y_i} = x_i^T\theta
\end{equation}
The weight matrix $\theta\in \mathbb{R}^{d \times K}$ is typically trained to minimize a negative cross-entropy loss function over all datapoints: 
\begin{equation}
\min_{\theta} \mathcal{L}(\theta) = -\frac{1}{N}\sum_{i=1}^N \sum_{k=1}^K P_{ik}\log \hat{P}_{ik}(\theta) \label{cross_entropy_loss}
\end{equation}
Here, $P_{ik}$ is the true probability that a datapoint $x_i$ belongs to the category $k$. As data can only exist in one category, this probability is simply the indicator function $\mathbb{I}[y_i = k]$. The probability $\hat{P}_{ik}$ is the result of a softmax transformation on the model output $\hat{y} \in \mathbb{R}^K$, where low deterministic outcomes $\hat{y}$ are mapped to high probability scores. 

So far, we have outlined a standard multi-class classification problem. However, one can view these generated scores as parameters for the categorical likelihood distribution $p(y_i|\theta) = \text{Cat}(y_i|\hat{P}_{i1},\dots,\hat{P}_{iK}, \theta)$. Hence, minimizing the negative cross-entropy is equivalent to maximum likelihood estimation (MLE): 
\begin{equation}
\max_{\theta} \mathcal{L}(\theta)
= \frac{1}{N}\sum_{i=1}^N   \log \text{Cat}(y_i|\hat{P}_{i1},\dots,\hat{P}_{iK}, \theta)\label{cat_prob}
\end{equation}
The probabilistic ML frame encodes choices regarding the statistical independence between variables. Each observation $y_i$ is independent,conditioned on the weight parameters $\theta$ and the data $x_i$. As such, each conditional distribution $p(y_i|\theta)$ is independent from one another -- a bias revealed by the summation over all of the likelihood distributions after applying the $\log$ to their product. 

Once in the probabilistic frame, Bayesian methods enable one to take advantage of the entire joint distribution between variables $p(y,\theta)$. In the example above, we only value obtaining parameters for the factorized conditionals $\prod p(y_i|\theta)$. In doing this, we treat $\theta$ as a point estimate to optimize. Instead, one could choose to incorporate uncertainty over the parameters $p(\theta)$ into the modeling. If this prior is fixed, then the problem becomes one of finding the parameters for the $\log$-likelihood distribution that are possible under the \emph{constraint} of the prior $p(\theta)$. Intuitively, the prior is meant to constrain the space of possible point estimates for $\theta$ by incorporating \emph{known uncertainty}. 
The likelihood and prior terms interact according to Bayes' rule: 
\begin{equation}
p(\theta|y) = \frac{p(y|\theta)p(\theta)}{p(y)}\label{Bayes_rule}
\end{equation}
The term on the left of the equals sign is the \emph{posterior}, which quantifies the uncertainty over the parameters $\theta$ after observing all of the training data. The denominator on the right of the equals sign is the \emph{evidence} -- also referred to as the marginal likelihood-- which is the likelihood of observing the output data given all of the input data and summed over all latent parameter possibilities. We will come back to the utility of these terms shortly. 

One important note about Bayes' rule is its own commutative structure. Swapping the order of the variables in a joint distribution (the denominator in Bayes rule) offers the same result: 
\begin{equation}
p(x,y) = {p(x|y)p(y)} = {p(y|x)p(x)}
\end{equation}
Thus, order structure is not baked into Bayesian ML methods inherently. To account for order, one requires additional variables in the probabilistic model or a graphical base model that incorporate an arrow of causality. 

Notice that a probabilistic framing is compatible, but separate from statistical learning theory (SLT), where minimizing the negative cross-entropy in Eq.~\ref{cross_entropy_loss} is akin to computing the empirical risk that approximates the true risk in expectation: 
\begin{equation}
\text{True risk}
= \mathbb{E}_{x,y\sim p(x,y)}[\mathcal{L}(\theta)]\label{true_risk}
\end{equation}
This theoretical assumption typically underlies standard generalization bounds for a model class. Consider an ML practitioner who has a training dataset consisting of images with categorical labels (e.g.~dogs, cats, bunnies, horses). However, this practitioner has very few images of bunnies (bunnies are quick and hard to take pictures of!!) compared to the rest of the classes. In practice, we refer to this phenomena as \emph{class imbalance}. To fix this imbalance, a practitioner might up-weight the minority class predictions in the cost function in Eq.~\ref{cross_entropy_loss}. Adding this term requires one to update Eq.~\ref{true_risk} to a weighted true risk; otherwise the empirical risk would be an incorrect estimator. Thus, in order to understand practical performance limits, the theory needs to reflect the choices that practitioners make to model the data well. It quickly becomes a challenge for theorists to bake in the rapidly changing design choices that are necessary to account for the interaction between ML models and the properties of real-world data. In addition, generalization bounds often characterize the resources required for all functions in a given model class to achieve a specified generalization error. As a result, they primarily describe worst-case requirements that do not take into account 
alignment between the model and the data. We include this brief comment not to discourage researchers from pursuing SLT, but rather to distinguish SLT from the probabilistic framing in ML. SLT offers a framework for analyzing model classes with general data assumptions, whereas the probabilistic framing offers a workspace for principled model design and iteration.

In summary, viewing ML problems through a probabilistic framing allows for designing models with inherent interpretability: one can explicitly define the variables within the problem, choose independence structure to reflect real-world conditions, and incorporate known uncertainties that are of value. One can also choose probabilistic function families that align with the domain of the problem (e.g.~binary classification: Bernoulli, multi-class: categorical, circular dynamics: von Mises). Bayesian methods offer clever ways in which inductive biases can further be woven into probability distributions. Crucial to practical implementation are the estimation and optimization methods for the posterior and evidence terms. We explore this next, and briefly discuss how some assumptions made for computational tractability bias the model behavior. 

\subsection{Posterior estimation for prediction}

As stated previously, the posterior in Eq.~\ref{Bayes_rule} reveals the uncertainty over the base model parameters, given all of the training observations. Our aim is to estimate this posterior, so that it can be used as the ``updated prior" when making predictions on unseen data. For certain likelihood and prior distribution families, the posterior can be computed analytically (e.g.~Gaussian distributions). However, often times this posterior cannot be computed with an analytical form and exact numerical evaluations are intractable (especially when the base model is a $>1$ billion parameter neural network). Rather, one approximates the posterior using the proportionality condition $p(\theta|x,y) \propto p(y|\theta,x)p(\theta)$. While there are likely a plethora of methods out there, we mention two popular choices: 

\noindent \textbf{MAP + Laplace.} Maximum a posteriori (MAP) estimation is MLE under a prior. Rather than only maximizing the $\log-$likelihood, one maximizes the following: 
\begin{equation}
\max_{\theta} \mathcal{L}(\theta)
= \frac{1}{N}\sum_{i=1}^N   [\log p(y_i|\theta)] + \log p(\theta) 
\end{equation}
In practice, we minimize the negative log-likelihood. In doing so, one obtains an optimal point estimate $\theta^*$ for the negative $\log$ posterior $-\log p(\theta^*|y)$. We then estimate the shape of the posterior function via a Laplace approximation, which simply means that we use the second-order term from the Taylor expansion to approximate local curvature around the optimal point. The second-order term maps directly to the quadratic term in the multi-variate Gaussian distribution: 
\begin{equation}
\begin{aligned}
\text{Gaussian exponential term:} 
&\quad -\frac{1}{2}[(y - \mu)^T\Sigma^{-1}(y-\mu)] \\
\text{Second-order Taylor term:} 
&\quad -\frac{1}{2}[(\theta - \theta^*)^T H(\theta - \theta^*)]
\end{aligned}
\end{equation}
This reveals that the Hessian $H$ (second derivative matrix) $\nabla_{\theta^*}^2 [-\log p(\theta^*|y)]$ operates as the precision matrix $\Sigma^{-1}$. We invert the Hessian matrix to obtain the covariance matrix for the posterior. This computation is again expensive when the number of parameters $\theta$ is large, such that further assumptions on the Hessian matrix are typically made for tractability (e.g.~diagonal). Note that these further assumptions introduce further \emph{bias} into the problem; however it is often less studied how these biases influence the posterior. 

\noindent \textbf{Variational inference (VI).} An oftentimes more tractable method to estimate the posterior is to approximate it with an assumed distribution and optimize the parameters of this distribution directly. This requires one to make assumptions regarding the posterior distribution family and its factorization -- i.e.~which posterior variables are independent. For example, an assumed posterior can be a Gaussian distribution with an isotropic diagonal covariance, which implies that each model parameter (an entry in the matrix $\theta$) is independent of one another and has the same variance. Again, these assumptions introduce further bias into the problem, which has more recently been recognized with initial attempts to study the resulting behavior\cite{harvey2026learning, harvey2026occams}.  

Estimating the posterior through optimization requires one to maximize the \emph{Evidence Lower Bound} (ELBO): 
\begin{equation}
\begin{aligned}
\mathrm{ELBO}(\theta, \phi)
&=
\sum_{i=1}^N
\mathbb{E}_{\theta \sim q_{\phi}}\log 
\left[
p(y_i\mid\theta)
\right] \\
&\quad
-
\mathrm{KL}
\left(q_{\phi}(\theta)
\,\middle\|\,
p(\theta)
\right)
\end{aligned}
\label{ELBO}
\end{equation}
Here, the assumed posterior $q_{\phi}$ contains distribution parameters $\phi$ to be optimized alongside the base model parameters $\theta$. Notice that the first term of the ELBO is simply the data-likelihood, but averaged over parameters $\theta$ that are sampled from the assumed posterior. Thus, the optimization encourages a model that maximizes likelihood under parameters $\theta$ that are sampled from the posterior. The second term introduces a constraint from the prior, which biases the posterior to be similar to the prior. The KL represents the standard Kullback Libeler (KL) divergence between two distributions, indicating their similarity. Optimizing the ELBO parameters is akin to optimizing a posterior distribution directly. One can up-weight or down-weight the data-likelihood or prior term depending on the ML context. 

\subsection{Bayesian evidence for hyper-parameter optimization}\label{optimization}
The evidence term in Eq.~\ref{Bayes_rule} is most useful when one is in search of a \emph{good} probabilistic model, where good is defined by Occam's razor -- a simple model that fits the data well. The evidence is typically written as the integral over the model parameters with respect to the data:  
\begin{equation}
p(y) = \int p(y|\theta) p(\theta)d\theta\label{evidence}
\end{equation}
The evidence quantifies the quality of the model's ability to capture the data well under the simplicity of the prior parameters. Intuitively one can see that if the prior has a non-zero density over too parameter combinations, the integral is a sum over many smaller terms -- since we have the constraint $\int p(\theta)d\theta=1$. If the chosen parameters produce a small likelihood term, this also penalizes the evidence to be small. Thus, the highest evidence for a set of data is the model that achieves Occam's razor. The evidence can be used as a metric to compare probabilistic models for a specific dataset, and thus can be used as an optimization tool. 

Optimizing the evidence (or the $\log$ evidence) is heavily linked to base model hyper-parameter optimization that eliminates the need for a validation set. One typically uses a validation set to pick a base model containing hyper-parameters that perform best on unseen data. If the base model hyper-parameters can viewed as probabilistic parameters and one uses the evidence as the metric for a \emph{good model}, then one can bypass training many models via a grid search over hyper-parameters by optimizing the evidence directly. Since the evidence term is typically intractable to compute, one can use the ELBO in Eq.~\ref{ELBO} to optimize the evidence indirectly.

The ELBO is derived from lower bounding the $log$ evidence in the following way: 
\begin{flalign}
&\text{Evidence:}\quad
\log p(y )
= \log \int 
\frac{q_{\phi}(\theta)}{q_{\phi}(\theta)}
p(y, \theta) \, d\theta && \notag \\
&\hspace{5.8em}
= \log \mathbb{E}_{\theta \sim q_{\phi}}
\left[
\frac{p(y, \theta)}{q_{\phi}(\theta)}
\right], && \notag \\[0.75em]
&\text{Jensen's inequality:}\quad
\mathbb{E}[\log f(x)]
\leq \log \mathbb{E}[f(x)], && \notag \\[0.75em]
&\text{Optimize:}\quad
\mathbb{E}_{\theta \sim q_{\phi}}
\left[
\log \frac{p(y, \theta)}{q_{\phi}(\theta)}
\right] &&
\end{flalign}

Another way to view the ELBO is a maximization of the joint distribution between output data and model parameters, while maximizing the entropy of the assumed posterior. Hence, we directly optimize for structure while attempting to use the expressive power of the posterior distribution. The chosen posterior distribution family is going to heavily influence this expressive power. 

Optimizing the evidence indirectly through VI to learn a well-regularized model is a technique used in the Gaussian process kernel design process, discussed in Sec.~\ref{Bayesian_GPs}. The parameters of the frequency distribution chosen to match a target kernel function can be treated as model hyper-parameters \cite{harvey2026learning}, or the frequencies themselves can be learned as hyper-parameters given an initialization biased by a chosen kernel \cite{lazaro-gredilla2010sparse}. 

\subsection{Exact GP inference}\label{exact_inference}

A GP is defined as a Gaussian distribution over functions, $f(x) \sim \mathcal {GP}(\mu, k(x_i, x_i'))$, where $\mu$ is a vector of mean function evaluations and the covariance matrix is determined by pairwise input similarities encoded by the kernel function $k(x_i,x_i')$. In exact GP inference, one typically starts by defining a GP prior over functions $f(x) \sim \mathcal GP(0_N, k_{xx'})$, and a corresponding likelihood $p(y|f(x))$. The $N \times N$ matrix $k_{xx'}$ contains information between all pairs of datapoints. We show the function-space posterior predictive for when the likelihood is Gaussian $\mathcal{N}(f(x), I\sigma^2_{\text{noise}})$. For a new test point $x_i^{\text{test}}$, one computes the mean and variance for the posterior predictive $p(f(x_i^{\text{test}})|y)$ directly using the following closed-form expressions: 
\begin{equation}
\begin{aligned}
&\text{Exact GP posterior predictive:} \\
&p\!\left(f(x_i^{\text{test}}) \mid y\right)
=
\mathcal{N}
\left(
\mu_{\mathrm{pred}},
\sigma^2_{\mathrm{pred}}
\right), \\[0.5em]
&\mu_{\mathrm{pred}}
=
k_{x^{\text{test}}x}
\left(
k_{xx} + \sigma^2_{\mathrm{noise}} I
\right)^{-1}
y, \\[0.5em]
&\sigma^2_{\mathrm{pred}}
=
k_{x^{\text{test}}x^{\text{test}}}
-
k_{x^{\text{test}}x}
\left(
k_{xx} + \sigma^2_{\mathrm{noise}} I
\right)^{-1}
k_{xx^{\text{test}}}
\end{aligned}
\label{eq:posterior_predictive}
\end{equation}
\normalsize

Here, $k_{x^{\text{test}}x}$ is a $N-$dimensional vector containing the overlaps of the training data with the new test point. These values are numerically equivalent to those in the weight-space view in Eq.~\ref{posterior_predictive} when the kernel is $k_{xx} = \phi(x)\phi(x)^T$, but are written in a different form. 
The inversion of the data-space matrix $k_{xx}$ is the computational bottleneck for exact GP inference - regardless of whether the feature map is a Fourier embedding or an identity function (i.e.~no feature embedding in the weight-space view).

\section{A soft introduction to the representation theory of groups}\label{app:rep-theory}
This appendix serves as an introduction to the representation theory of groups for those unfamiliar with the basic concepts. We begin with an understanding of the notion of a group, which is intuitively thought of as a collection of symmetries of some kind. Symmetries typically have several nice properties, which we illustrate with the example of the symmetries of a square. First, combining two symmetries should yield another symmetry; a 90 degree rotation of a square is a symmetry of the square, and composing this rotation with itself results in another symmetry, namely the 180 degree rotation of the square. Second, doing nothing is a symmetry; leaving the square alone leaves it unchanged. Third, symmetries are invertible by another symmetry; rotating the square by 90 degrees is a symmetry, and rotating by -90 degrees is also a symmetry that undoes the first. Finally, symmetries are typically associative; if $r$ denotes the rotation of the square and $f$ a flip along some symmetry axis, we have $r(fr)=(rf)r$. This last property is not to be confused with commutativity, which equates all orderings of the composition of many symmetries. Most groups are not commutative (also known as abelian). Indeed, in our square example, we have $rf\ne fr$.

Explicitly, a group is a set $G$ equipped with an associative binary operation from $G\times G$ to $G$ that contains an identity element and an inverse for every one of its elements. The binary operation is sometimes denoted $g*h$, but typically the $*$ symbol is dropped in favor of juxtaposition, so that the same binary operation is denoted $gh$ for $g,h\in G$. The requirement that an identity exist means that there is an element $1\in G$ such that $1g=g1=g$ for all $g\in G$. The requirement that inverses exist means that for every $g\in G$, there is a $g^{-1}\in G$ such that $gg^{-1}=1=g^{-1}g$. The requirement that the binary operation be associative means that $g(hk)=(gh)k$ for all $g,h,k\in G$. Thus, the definition of a group captures precisely the properties of a collection of symmetries of an object outlined above.

Sometimes groups are essentially equivalent in the sense that they have exactly the same structure up to a relabeling of the elements. In this case, we call the groups isomorphic. Let us capture this notion of essential equivalence mathematically. The two properties unique to a group are its set and binary operation structures. Consider two groups $(G,*_G)$ and $(H,*_H)$. A homomorphism is a map $\varphi:G\to H$ such that $\varphi(g_1*_G g_2)=\varphi(g_1)*_H\varphi(g_2)$. Such a map takes the binary operation in $G$ and preserves it in $H$; composing two elements of $G$ and then mapping the result to $H$ is the same as mapping two elements to $H$ and then composing them in $H$. Thus, the elements of the image of $G$ under a homomorphism are forced to behave in the same way they did in $G$, but now viewed as elements of $H$. If in addition to being a homomorphism, $\varphi$ is a bijection (equivalently, invertible), then at the set level, it is just a relabeling of the elements of $G$ with the elements of $H$, thereby preserving both the binary operation and the set structure. Such a map is called an isomorphism, and two groups that allow an isomorphism between them are called isomorphic.

The definition of a group can be fairly abstract, but an isomorphism can carry the same structure from the abstract setting to a concrete one. Even when the map is not an isomorphism but a homomorphism, it still preserves the structure of the binary operation on the group, and the image of the map can be identified with a subgroup of the co-domain. In fact, the first isomorphism theorem guarantees that $\varphi(G)\cong G/\ker(\varphi)$, where $\cong$ denotes the isomorphism property. Thus, faithfully representing $G$ as a subgroup of $H$ amounts to finding a homomorphism $\varphi:G\to H$ with a trivial kernel. For this reason, such a map is sometimes called faithful. This viewpoint motivates the definition of a representation. In quantum computation, we work largely in the group $\mathcal{U}(\mathcal{H})$ consisting of all unitary operations on some Hilbert space $\mathcal{H}$. We would therefore like a way to understand symmetry with respect to $G$ in the context of its unitary action on $\mathcal{H}$. We do so by finding a faithful homomorphism $\varphi:G\to\mathcal{U}(\mathcal{H})$, so that $\varphi(G)\cong G$, allowing us to think about $G$ as a subgroup of the unitary group. Such a map is called a faithful unitary representation of $G$.

A faithful unitary representation $\varphi:G\to\mathcal{U}(\mathcal{H})$ transfers the abstract notion of symmetry in $G$ to the concrete setting of unitary operators acting on a Hilbert space. It allows us to talk about states $\ket{\psi}\in\mathcal{H}$ which are invariant under the action of $G$: $\varphi(G)\ket{\psi}=\ket{\psi}$. We therefore think of invariant states as satisfying the symmetry in $G$. It should be noted, however, that invariance depends on the chosen representation $\varphi$. It is a common abuse of notation to say $\ket{\psi}$ is $G$-invariant, when we really mean $\varphi(G)$-invariant. 

In the setting of quantum machine learning, this is a valuable construction because our data is usually classical before getting encoded into a Hilbert space. The encoder therefore needs to transfer the notion of symmetry from the data space to the Hilbert space in order for invariance properties to be subsumed. However, the encoder is not itself a map between the groups, but a map between the spaces carrying the representations in the data and Hilbert space pictures. Indeed, the encoder has the form $\mathcal{E}:\mathcal{X}\to\mathcal{H}$. The symmetry in the data picture is given by some natural action of the group $G$ on the set $\mathcal{X}$, which we denote $g\cdot x$. If $R:G\to\mathcal{U}(\mathcal{H})$ is a faithful unitary representation of $G$ acting on a Hilbert space, then the notion of $G$-invariance in $\mathcal{X}$ is carried to a notion of $G$-invariance in $R(G)\subset\mathcal{U}(\mathcal{H})$ by the encoding if
\begin{equation}
    \mathcal{E}(g\cdot x)=R(g)\mathcal{E}(x).
\end{equation}
Such an encoding is called an equivariant map. Intuitively, this equation says that invariance before mapping $x$ to $\mathcal{E}(x)$ implies invariance after the mapping. The converse is also true precisely when the encoder is injective.

\bibliography{verified_refs.bib}

@article{Preskill2018,
  author = {John Preskill},
  title = {Quantum Computing in the {NISQ} era and beyond},
  journal = {Quantum},
  volume = {2},
  pages = {79},
  year = {2018},
  publisher = {Verein zur Forderung des Open Access Publizierens in den Quantenwissenschaften},
  month = {aug},
  doi = {10.22331/q-2018-08-06-79},
}

@article{Bharti2022_nisqreview,
  author = {Bharti, Kishor and Cervera-Lierta, Alba and Kyaw, Thi Ha and Haug, Tobias and Alperin-Lea, Sumner and Anand, Abhinav and Degroote, Matthias and Heimonen, Hermanni and Kottmann, Jakob S. and Menke, Tim and others},
  title = {Noisy intermediate-scale quantum algorithms},
  journal = {Reviews of Modern Physics},
  volume = {94},
  pages = {015004},
  year = {2022},
  publisher = {American Physical Society (APS)},
  month = {Feb},
  doi = {10.1103/RevModPhys.94.015004},
  number = {1},
}

@article{zanetta2026efficient,
  title   = {Efficient Modeling of Subkilometer Surface Wind
             with {Gaussian} Processes and Neural Networks},
  author  = {Zanetta, Francesco and Nerini, Daniele
             and Buzzi, Matteo and Moss, Henry},
  journal = {Artificial Intelligence for the Earth Systems},
  volume  = {5},
  number  = {2},
  year    = {2026},
  doi     = {10.1175/AIES-D-24-0047.1},
  url     = {https://doi.org/10.1175/AIES-D-24-0047.1}
}

@inproceedings{borovitskiy2020matern,
  title     = {{Mat{\'e}rn Gaussian} Processes on
               {Riemannian} Manifolds},
  author    = {Borovitskiy, Viacheslav and Terenin, Alexander
               and Mostowsky, Peter and Deisenroth, Marc Peter},
  booktitle = {Advances in Neural Information Processing Systems},
  volume    = {33},
  year      = {2020},
  url       = {https://proceedings.neurips.cc/paper/2020/hash/92bf5e6240737e0326ea59846a83e076-Abstract.html}
}

@inproceedings{chau2021deconditional,
  title     = {Deconditional Downscaling with {Gaussian} Processes},
  author    = {Chau, Siu Lun and Bouabid, Shahine
               and Sejdinovic, Dino},
  booktitle = {Advances in Neural Information Processing Systems},
  volume    = {34},
  pages     = {17813--17825},
  year      = {2021},
  url       = {https://proceedings.neurips.cc/paper/2021/hash/94aef38441efa3380a3bed3faf1f9d5d-Abstract.html}
}

@misc{wang2026explainablequantumregressionalgorithm,
      title={Explainable quantum regression algorithm with encoded data structure}, 
      author={C. -C. Joseph Wang and F. Perkkola and I. Salmenperä and A. Meijer-van de Griend and J. K. Nurminen},
      year={2026},
      eprint={2604.15666},
      archivePrefix={arXiv},
      primaryClass={quant-ph},
      url={https://arxiv.org/abs/2604.15666}, 
}

@misc{rodriguezgrasa2026pacbayesianapproachgeneralizationquantum,
      title={A PAC-Bayesian approach to generalization for quantum models}, 
      author={Pablo Rodriguez-Grasa and Matthias C. Caro and Jens Eisert and Elies Gil-Fuster and Franz J. Schreiber and Carlos Bravo-Prieto},
      year={2026},
      eprint={2603.22964},
      archivePrefix={arXiv},
      primaryClass={quant-ph},
      url={https://arxiv.org/abs/2603.22964}, 
}

@article{Cerezo2021_vqareview,
  author = {Cerezo, M. and Arrasmith, Andrew and Babbush, Ryan and Benjamin, Simon C. and Endo, Suguru and Fujii, Keisuke and McClean, Jarrod R. and Mitarai, Kosuke and Yuan, Xiao and Cincio, Lukasz and others},
  title = {Variational quantum algorithms},
  journal = {Nature Reviews Physics},
  volume = {3},
  number = {9},
  pages = {625--644},
  year = {2020},
  doi = {10.1038/s42254-021-00348-9},
}

@article{Moll2017,
  author = {Nikolaj Moll and Panagiotis Barkoutsos and Lev S Bishop and Jerry M Chow and Andrew Cross and Daniel J Egger and Stefan Filipp and Andreas Fuhrer and Jay M Gambetta and Marc Ganzhorn and others},
  title = {Quantum optimization using variational algorithms on near-term quantum devices},
  journal = {Quantum Science and Technology},
  volume = {3},
  number = {3},
  pages = {030503},
  year = {2017},
  doi = {10.1088/2058-9565/aab822},
}

@article{Wecker_PRA2015,
  author = {Wecker, Dave and Hastings, Matthew B. and Troyer, Matthias},
  title = {Progress towards practical quantum variational algorithms},
  journal = {Physical Review A},
  volume = {92},
  pages = {042303},
  year = {2015},
  publisher = {American Physical Society (APS)},
  month = {Oct},
  doi = {10.1103/PhysRevA.92.042303},
  number = {4},
}

@article{lubasch2020variational,
  author = {Lubasch, Michael and Joo, Jaewoo and Moinier, Pierre and Kiffner, Martin and Jaksch, Dieter},
  title = {Variational quantum algorithms for nonlinear problems},
  journal = {Physical Review A},
  volume = {101},
  number = {1},
  pages = {010301},
  year = {2019},
  publisher = {APS},
  doi = {10.1103/PhysRevA.101.010301},
}

@article{Farhi2018,
  author = {Edward Farhi and Hartmut Neven},
  title = {Classification with Quantum Neural Networks on Near Term Processors},
  journal = {arXiv:1802.06002},
  year = {2018},
  url = {https://arxiv.org/abs/1802.06002},
  doi = {10.37686/qrl.v1i2.80},
}

@article{cong2019quantumconvolutional,
  author = {Cong, Iris and Choi, Soonwon and Lukin, Mikhail D.},
  title = {Quantum Convolutional Neural Networks},
  journal = {Nature Physics},
  volume = {15},
  pages = {1273--1278},
  year = {2018},
  doi = {10.1038/s41567-019-0648-8},
}

@inproceedings{bausch2020recurrent,
  author = {Bausch, Johannes},
  title = {Recurrent Quantum Neural Networks},
  booktitle = {Neural Information Processing Systems},
  volume = {33},
  year = {2020},
  publisher = {Curran Associates, Inc.},
  url = {https://proceedings.neurips.cc/paper/2020/hash/0ec96be397dd6d3cf2fecb4a2d627c1c-Abstract.html},
  journal = {Neural Information Processing Systems},
}

@article{Benedetti2021,
  author = {Benedetti, Marcello and Coyle, Brian and Fiorentini, Mattia and Lubasch, Michael and Rosenkranz, Matthias},
  title = {Variational Inference with a Quantum Computer},
  journal = {Physical Review Applied},
  volume = {16},
  pages = {044057},
  year = {2021},
  publisher = {American Physical Society (APS)},
  month = {Oct},
  doi = {10.1103/PhysRevApplied.16.044057},
  number = {4},
}

@article{NonlinearQCBMs,
  author = {Gili, Kaitlin and Sveistrys, Mykolas and Ballance, Chris},
  title = {Introducing nonlinear activations into quantum generative models},
  journal = {Physical Review A},
  volume = {107},
  number = {1},
  pages = {012406},
  year = {2022},
  publisher = {American Physical Society},
  doi = {10.1103/PhysRevA.107.012406},
}

@inproceedings{chen2022quantumlong,
  author = {Chen, Samuel Yen-Chi and Yoo, Shinjae and Fang, Yao-Lung L.},
  title = {Quantum Long Short-Term Memory},
  booktitle = {IEEE International Conference on Acoustics, Speech, and Signal Processing},
  pages = {8622--8626},
  year = {2020},
  organization = {IEEE},
  doi = {10.1109/ICASSP43922.2022.9747369},
  journal = {IEEE International Conference on Acoustics, Speech, and Signal Processing},
}

@article{cherrat2024quantumvision,
  author = {Cherrat, El Amine and Kerenidis, Iordanis and Mathur, Natansh and Landman, Jonas and Strahm, Martin and Li, Yun Yvonna},
  title = {Quantum {V}ision {T}ransformers},
  journal = {Quantum},
  volume = {8},
  pages = {1265},
  year = {2022},
  doi = {10.22331/q-2024-02-22-1265},
}

@article{abbas2021power,
  author = {Abbas, Amira and Sutter, David and Zoufal, Christa and Lucchi, Aur{\'e}lien and Figalli, Alessio and Woerner, Stefan},
  title = {The power of quantum neural networks},
  journal = {Nature Computational Science},
  volume = {1},
  number = {6},
  pages = {403--409},
  year = {2020},
  publisher = {Nature Publishing Group},
  doi = {10.1038/s43588-021-00084-1},
}

@article{Gili2022,
  author = {Gili, Kaitlin and Hibat-Allah, Mohamed and Mauri, Marta and Ballance, Chris and Perdomo-Ortiz, Alejandro},
  title = {Do Quantum Circuit Born Machines Generalize?},
  journal = {Quantum Science and Technology},
  volume = {8},
  number = {3},
  pages = {035021},
  year = {2022},
  doi = {10.1088/2058-9565/acd578},
}

@article{hur2022quantumconvolutional,
  author = {Hur, Tak and Kim, Leeseok and Park, Daniel K.},
  title = {Quantum Convolutional Neural Network for Classical Data Classification},
  journal = {Quantum Machine Intelligence},
  volume = {4},
  number = {1},
  pages = {3},
  year = {2021},
  doi = {10.1007/s42484-021-00061-x},
}

@article{khoo2024benchmarking,
  author = {Khoo, Jun Yong and Gan, Chee Kwan and Ding, Wenjun and Carrazza, Stefano and Ye, Jun and Kong, Jian Feng},
  title = {Benchmarking Quantum Convolutional Neural Networks for Classification and Data Compression Tasks},
  journal = {arXiv preprint arXiv:2411.13468},
  year = {2024},
  url = {https://arxiv.org/abs/2411.13468},
}

@article{bowles2024better,
  author = {Bowles, Joseph and Ahmed, Shahnawaz and Schuld, Maria},
  title = {Better than classical? The subtle art of benchmarking quantum machine learning models},
  journal = {arXiv.org},
  year = {2024},
  url = {https://arxiv.org/abs/2403.07059},
  doi = {10.48550/arXiv.2403.07059},
}

@article{basilewitsch2025quantum,
  author = {Basilewitsch, Daniel and Bravo, Jo{\~a}o F. and Tutschku, Christian and Struckmeier, Frederick},
  title = {Quantum Neural Networks in Practice: A Comparative Study with Classical Models from Standard Data Sets to Industrial Images},
  journal = {Quantum Machine Intelligence},
  volume = {7},
  pages = {110},
  year = {2024},
  doi = {10.1007/s42484-025-00336-7},
}

@article{qmetric2025benchmarking,
  author = {Illésová, Silvie and Rybotycki, Tomasz and Beseda, Martin},
  title = {QMetric: Benchmarking Quantum Neural Networks Across Circuits, Features, and Training Dimensions},
  journal = {QualITA},
  year = {2025},
  url = {https://arxiv.org/abs/2506.23765},
}

@article{schuld2019evaluating,
  author = {Schuld, Maria and Bergholm, Ville and Gogolin, Christian and Izaac, Josh and Killoran, Nathan},
  title = {Evaluating analytic gradients on quantum hardware},
  journal = {Physical Review A},
  volume = {99},
  number = {3},
  pages = {032331},
  year = {2018},
  publisher = {APS},
  doi = {10.1103/PhysRevA.99.032331},
}

@article{Mcclear2018Barren,
  author = {Mcclean, Jarrod and Boixo, Sergio and Smelyanskiy, Vadim and Babbush, Ryan and Neven, Hartmut},
  title = {Barren plateaus in quantum neural network training landscapes},
  journal = {Nature Communications},
  volume = {9},
  year = {2018},
  month = {11},
  doi = {10.1038/s41467-018-07090-4},
  number = {1},
  publisher = {Springer Science and Business Media LLC},
}

@article{holmes2021scramblers,
  author = {Holmes, Zo{\"e} and Arrasmith, Andrew and Yan, Bin and Coles, Patrick J and Albrecht, Andreas and Sornborger, Andrew T},
  title = {Barren plateaus preclude learning scramblers},
  journal = {Physical Review Letters},
  volume = {126},
  number = {19},
  pages = {190501},
  year = {2020},
  publisher = {APS},
  doi = {10.1103/PhysRevLett.126.190501},
}

@article{wang2021noiseinduced,
  author = {Wang, Samson and Fontana, Enrico and Cerezo, Marco and Sharma, Kunal and Sone, Akira and Cincio, Lukasz and Coles, Patrick J},
  title = {Noise-induced barren plateaus in variational quantum algorithms},
  journal = {Nature communications},
  volume = {12},
  number = {1},
  pages = {1--11},
  year = {2020},
  publisher = {Nature Publishing Group},
  doi = {10.1038/s41467-021-27045-6},
}

@article{holmes2022expressivity,
  author = {Holmes, Zo{\"e} and Sharma, Kunal and Cerezo, Marco and Coles, Patrick J},
  title = {Connecting ansatz expressibility to gradient magnitudes and barren plateaus},
  journal = {PRX Quantum},
  volume = {3},
  number = {1},
  pages = {010313},
  year = {2021},
  publisher = {APS},
  doi = {10.1103/PRXQuantum.3.010313},
}

@article{cerezo2021costfunction,
  author = {Cerezo, Marco and Sone, Akira and Volkoff, Tyler and Cincio, Lukasz and Coles, Patrick J},
  title = {Cost function dependent barren plateaus in shallow parametrized quantum circuits},
  journal = {Nature communications},
  volume = {12},
  number = {1},
  pages = {1--12},
  year = {2021},
  publisher = {Springer Science and Business Media LLC},
  doi = {10.1038/s41467-021-21728-w},
}

@article{arrasmith2021gorges,
  author = {Arrasmith, Andrew and Holmes, Zo{\"e} and Cerezo, Marco and Coles, Patrick J.},
  title = {Equivalence of Quantum Barren Plateaus to Cost Concentration and Narrow Gorges},
  journal = {Quantum Science and Technology},
  volume = {7},
  number = {4},
  pages = {045015},
  year = {2021},
  doi = {10.1088/2058-9565/ac7d06},
}

@article{larocca2025review,
  author = {Larocca, Martín and Thanasilp, Supanut and Wang, Samson and Sharma, Kunal and Biamonte, Jacob and Coles, Patrick J. and Cincio, Lukasz and McClean, Jarrod R. and Holmes, Zoë and Cerezo, M.},
  title = {A review of barren plateaus in variational quantum computing},
  journal = {Nature Reviews Physics},
  volume = {7},
  pages = {174--189},
  year = {2025},
  doi = {10.1038/s42254-025-00813-9},
  number = {4},
  publisher = {Springer Science and Business Media LLC},
}

@article{shi2025avoiding,
  author = {Shi, Xiao and Shang, Yun},
  title = {Avoiding barren plateaus via Gaussian mixture model},
  journal = {New Journal of Physics},
  volume = {27},
  number = {10},
  pages = {104501},
  year = {2024},
  publisher = {IOP Publishing},
  doi = {10.1088/1367-2630/ae0823},
}

@book{Bengio-Book,
  author = {Hao, Xingbang and Zhang, Guigang and Ma, Shang},
  title = {Deep Learning},
  year = {2016},
  publisher = {MIT Press},
  url = {http://www.deeplearningbook.org},
  journal = {International Journal of Semantic Computing},
  doi = {10.1142/S1793351X16500045},
}

@article{bengio2021deep,
  author = {Bengio, Yoshua and LeCun, Yann and Hinton, Geoffrey},
  title = {Deep learning for AI},
  journal = {Communications of the ACM},
  volume = {64},
  number = {7},
  pages = {58--65},
  year = {2021},
  doi = {10.1145/3448250},
  publisher = {Association for Computing Machinery (ACM)},
}

@article{paleyes2022challenges,
  author = {Paleyes, Andrei and Urma, Raoul-Gabriel and Lawrence, Neil D.},
  title = {Challenges in Deploying Machine Learning: A Survey of Case Studies},
  journal = {ACM Computing Surveys},
  volume = {55},
  number = {6},
  pages = {1--29},
  year = {2020},
  publisher = {ACM},
  doi = {10.1145/3533378},
}

@article{hendrickx2024machine,
  author = {Hendrickx, Kilian and Perini, Lorenzo and Plas, Dries Van der and Meert, Wannes and Davis, Jesse},
  title = {Machine learning with a reject option: a survey},
  journal = {Machine-mediated learning},
  year = {2021},
  doi = {10.1007/s10994-024-06534-x},
}

@article{hasan2023survey,
  author = {Hasan, M. and Abdar, M. and Khosravi, Abbas and Aickelin, U. and Lio’, Pietro and Hossain, Ibrahim and Rahman, Ashikur and Nahavandi, Saeid},
  title = {Survey on Leveraging Uncertainty Estimation Towards Trustworthy Deep Neural Networks: The Case of Reject Option and Post-training Processing},
  journal = {ACM Computing Surveys},
  volume = {57},
  number = {9},
  pages = {236},
  year = {2023},
  doi = {10.1145/3727633},
}

@article{norori2021addressing,
  author = {Norori, Natalia and Hu, Qiyang and Aellen, Florence M. and Faraci, Francesca D. and Tzovara, Athina},
  title = {Addressing bias in big data and AI for health care: A call for open science},
  journal = {Patterns},
  volume = {2},
  number = {10},
  pages = {100347},
  year = {2021},
  doi = {10.1016/j.patter.2021.100347},
  publisher = {Elsevier BV},
}

@article{turner2022race,
  author = {Turner, Brian E. and Steinberg, Joshua R. and Weeks, Brian T. and Rodriguez, Fernando and Cullen, Mark R.},
  title = {Race/ethnicity reporting and representation in US clinical trials: A cohort study},
  journal = {The Lancet Regional Health -- Americas},
  volume = {11},
  pages = {100252},
  year = {2022},
  doi = {10.1016/j.lana.2022.100252},
  publisher = {Elsevier BV},
}

@article{bereska2024mechanistic,
  author = {Bereska, Leonard and Gavves, Efstratios},
  title = {Mechanistic Interpretability for {AI} Safety -- A Review},
  journal = {Trans. Mach. Learn. Res.},
  year = {2024},
  url = {https://openreview.net/forum?id=ePUVetPKu6},
  doi = {10.48550/arXiv.2404.14082},
}

@article{somvanshi2026bridging,
  author = {Somvanshi, Shriyank and Islam, Md Monzurul and Rafe, Amir and Tusti, Anannya Ghosh and Chakraborty, Arka and Baitullah, Anika and Chowdhury, Tausif Islam and Alnawmasi, Nawaf and Dutta, Anandi and Das, Subasish},
  title = {Bridging the Black Box: A Survey on Mechanistic Interpretability in AI},
  journal = {ACM Computing Surveys},
  volume = {58},
  number = {8},
  year = {2026},
  doi = {10.1145/3787104},
  pages = {1-35},
  publisher = {Association for Computing Machinery (ACM)},
}

@article{Zschech2025Inherently,
  author = {Zschech, Patrick and Weinzierl, Sven and Kraus, Mathias},
  title = {Inherently Interpretable Machine Learning: A Contrasting Paradigm to Post-hoc Explainable AI},
  journal = {Business \& Information Systems Engineering},
  pages = {1--19},
  year = {2025},
  note = {Received 17 Dec 2024; accepted 24 Jul 2025; published 15 Sep 2025},
  doi = {10.1007/s12599-025-00964-0},
}

@article{Rudin2019StopExplaining,
  author = {Rudin, Cynthia},
  title = {Stop explaining black box machine learning models for high stakes decisions and use interpretable models instead},
  journal = {Nature Machine Intelligence},
  volume = {1},
  pages = {206--215},
  year = {2018},
  month = {may},
  note = {Received 30 Dec 2018; accepted 26 Mar 2019; published 13 May 2019},
  doi = {10.1038/s42256-019-0048-x},
}

@article{schuld2021effect,
  author = {Schuld, Maria and Sweke, Ryan and Meyer, Johannes Jakob},
  title = {The Effect of Data Encoding on the Expressive Power of Variational Quantum-Machine-Learning Models},
  journal = {Physical Review A},
  volume = {103},
  number = {3},
  pages = {032430},
  year = {2020},
  doi = {10.1103/PhysRevA.103.032430},
}

@inproceedings{rahimi2007random,
  author = {Rahimi, Ali and Recht, Benjamin},
  title = {Random Features for Large-Scale Kernel Machines},
  booktitle = {Neural Information Processing Systems},
  volume = {20},
  year = {2007},
  publisher = {Curran Associates, Inc.},
  url = {https://papers.nips.cc/paper/3182-random-features-for-large-scale-kernel-machines},
  journal = {Neural Information Processing Systems},
}

@book{rasmussen2006Gaussian,
  author = {Canelles, Gerard Martínez},
  title = {Gaussian Processes for Machine Learning},
  year = {2017},
  publisher = {MIT Press},
  address = {Cambridge, MA},
  url = {https://gaussianprocess.org/gpml/},
}

@misc{belis2026spectral,
  author = {Belis, Vasilis and Bowles, Joseph and Gupta, Rishabh and Peters, Evan and Schuld, Maria},
  title = {Spectral Methods: Crucial for Machine Learning, Natural for Quantum Computers?},
  year = {2026},
  url = {https://arxiv.org/abs/2603.24654},
  journal = {arXiv.org},
  doi = {10.48550/arXiv.2603.24654},
}

@misc{harvey2026occams,
  author = {Harvey, Ethan and Hughes, Michael C.},
  title = {Occam's Razor is Only as Sharp as Your ELBO},
  year = {2026},
  url = {https://arxiv.org/abs/2604.25984},
  journal = {arXiv.org},
  doi = {10.48550/arXiv.2604.25984},
}

@inproceedings{wilson2013Gaussian,
  author = {Wilson, Andrew Gordon and Adams, Ryan Prescott},
  title = {Gaussian Process Kernels for Pattern Discovery and Extrapolation},
  booktitle = {International Conference on Machine Learning},
  series = {Proceedings of Machine Learning Research},
  volume = {28},
  pages = {1067--1075},
  year = {2013},
  publisher = {PMLR},
  url = {https://proceedings.mlr.press/v28/wilson13.html},
  journal = {International Conference on Machine Learning},
}

@book{bishop2006pattern,
  author = {Strauss, Laura},
  title = {Pattern Recognition and Machine Learning},
  series = {Information Science and Statistics},
  year = {2016},
  publisher = {Springer},
  address = {New York, NY},
  doi = {10.1007/978-0-387-45528-0},
}

@article{lazaro-gredilla2010sparse,
  author = {Lázaro-Gredilla, M. and Candela, J. Q. and Rasmussen, C. and Figueiras-Vidal, A.},
  title = {Sparse Spectrum Gaussian Process Regression},
  journal = {Journal of Machine Learning Research},
  volume = {11},
  number = {63},
  pages = {1865--1881},
  year = {2010},
  url = {https://jmlr.org/papers/v11/lazaro-gredilla10a.html},
  doi = {10.5555/1756006.1859914},
}

@inproceedings{gal2015improving,
  author = {Gal, Yarin and Turner, Richard},
  editor = {Bach, Francis and Blei, David},
  title = {Improving the Gaussian Process Sparse Spectrum Approximation by Representing Uncertainty in Frequency Inputs},
  booktitle = {International Conference on Machine Learning},
  series = {Proceedings of Machine Learning Research},
  volume = {37},
  pages = {655--664},
  year = {2015},
  publisher = {PMLR},
  address = {Lille, France},
  url = {https://proceedings.mlr.press/v37/galb15.html},
  journal = {International Conference on Machine Learning},
}

@book{mackay2003information,
  author = {MacKay, David J. C.},
  title = {Information Theory, Inference, and Learning Algorithms},
  year = {2004},
  publisher = {Institute of Electrical and Electronics Engineers (IEEE)},
  address = {Cambridge, UK},
  url = {https://www.inference.org.uk/itprnn/book.pdf},
  journal = {IEEE Transactions on Information Theory},
  volume = {50},
  number = {10},
  pages = {2544-2545},
  doi = {10.1109/tit.2004.834752},
}

@article{liu2023simple,
  author = {Liu, Jeremiah Zhe and Padhy, Shreyas and Ren, Jie and Lin, Zi and Wen, Yeming and Jerfel, Ghassen and Nado, Zachary and Snoek, Jasper and Tran, Dustin and Lakshminarayanan, Balaji},
  title = {A Simple Approach to Improve Single-Model Deep Uncertainty via Distance-Awareness},
  journal = {Journal of Machine Learning Research},
  volume = {24},
  number = {42},
  pages = {1--63},
  year = {2022},
  url = {https://jmlr.org/papers/v24/22-0479.html},
  doi = {10.48550/arXiv.2205.00403},
}

@inproceedings{harvey2026learning,
  author = {Harvey, Ethan and Petrov, Mikhail and Hughes, Michael C.},
  title = {Learning Hyperparameters via a Data-Emphasized Variational Objective},
  booktitle = {arXiv.org},
  series = {Proceedings of Machine Learning Research},
  year = {2025},
  url = {https://openreview.net/forum?id=O2blXNfb0b},
  journal = {arXiv.org},
  doi = {10.48550/arXiv.2502.01861},
}

@article{strobl2025fourier,
  title   = {Fourier fingerprints of ansatzes in quantum machine learning},
  author  = {Strobl, Melvin and Sahin, M. Emre and van der Horst, Lucas and K{\"u}hn, Eileen and Streit, Achim and Jaderberg, Ben},
  journal = {Quantum Science and Technology},
  volume  = {11},
  number  = {4},
  pages   = {045008},
  year    = {2026},
  doi     = {10.1088/2058-9565/ae8df8},
  publisher = {IOP Publishing}
}

@article{oh2026fourier,
  author = {Oh, Seungcheol and Roh, Emily Jimin and Vasilakos, Athanasios V. and Park, Soohyun and Kim, Joongheon},
  title = {Fourier Analysis Perspective on Quantum Neural Networks},
  journal = {Communications Physics},
  volume = {9},
  number = {1},
  pages = {176},
  year = {2026},
  doi = {10.1038/s42005-026-02680-x},
  publisher = {Springer Science and Business Media LLC},
}

@misc{schuld2021supervised,
  author = {Schuld, Maria},
  title = {Supervised Quantum Machine Learning Models Are Kernel Methods},
  year = {2021},
  url = {https://arxiv.org/abs/2101.11020},
  journal = {Quantum Machine Intelligence},
  volume = {1},
  number = {3-4},
  pages = {65-71},
  doi = {10.1007/s42484-019-00007-4},
  publisher = {Springer Science and Business Media LLC},
}

@article{meyer2022exploiting,
  author = {Meyer, Johannes Jakob and Mularski, Marian and Gil-Fuster, Elies and Mele, Antonio Anna and Arzani, Francesco and Wilms, Alissa and Eisert, Jens},
  title = {Exploiting Symmetry in Variational Quantum Machine Learning},
  journal = {PRX Quantum},
  volume = {4},
  number = {1},
  pages = {010328},
  year = {2022},
  doi = {10.1103/PRXQuantum.4.010328},
}

@article{Larocca_2022,
  author = {Mart{\'i}n Larocca and Fr{\'{e}}d{\'{e}}ric Sauvage and Faris M. Sbahi and Guillaume Verdon and Patrick J. Coles and M. Cerezo},
  title = {Group-Invariant Quantum Machine Learning},
  journal = {{PRX} Quantum},
  volume = {3},
  number = {3},
  year = {2022},
  publisher = {American Physical Society ({APS})},
  month = {sep},
  doi = {10.1103/prxquantum.3.030341},
}

@misc{ragone2023representation,
  author = {Michael Ragone and Paolo Braccia and Quynh T. Nguyen and Louis Schatzki and Patrick J. Coles and Frederic Sauvage and Martin Larocca and M. Cerezo},
  title = {Representation Theory for Geometric Quantum Machine Learning},
  year = {2022},
  url = {https://arxiv.org/abs/2210.07980},
  journal = {arXiv.org},
  doi = {10.48550/arXiv.2210.07980},
}

@article{nguyen2024theory,
  author = {Nguyen, Quynh T. and Schatzki, Louis and Braccia, Paolo and Ragone, Michael and Coles, Patrick J. and Sauvage, Fr{\'e}d{\'e}ric and Larocca, Mart{\'i}n and Cerezo, M.},
  title = {Theory for Equivariant Quantum Neural Networks},
  journal = {PRX Quantum},
  volume = {5},
  number = {2},
  pages = {020328},
  year = {2022},
  publisher = {American Physical Society},
  doi = {10.1103/PRXQuantum.5.020328},
}

@article{stone1932one,
  author = {Stone, Marshall H},
  title = {On one-parameter unitary groups in Hilbert space},
  journal = {The Annals of Mathematics},
  volume = {33},
  number = {3},
  pages = {643},
  year = {1932},
  publisher = {JSTOR},
  doi = {10.2307/1968538},
}

@misc{bowles2023contextuality,
  author = {Joseph Bowles and Victoria J Wright and Máté Farkas and Nathan Killoran and Maria Schuld},
  title = {Contextuality and inductive bias in quantum machine learning},
  year = {2023},
  url = {https://arxiv.org/abs/2302.01365},
}

@article{bradshaw2026smoothness,
  title={What is Smoothness?},
  author={Bradshaw, Zachary P.},
  journal={arXiv preprint arXiv:2609.03246},
  year={2026},
  doi={10.48550/arXiv.2609.03246},
  url={https://arxiv.org/abs/2609.03246}
}

@article{gili2024inductive,
  author = {Gili, Kaitlin and Alonso, Guillermo and Schuld, Maria},
  title = {An inductive bias from quantum mechanics: learning order effects with non-commuting measurements},
  journal = {Quantum Machine Intelligence},
  volume = {6},
  pages = {67},
  year = {2023},
  doi = {10.1007/s42484-024-00200-0},
}

@book{searcoid2007metric,
  author = {Clason, Christian},
  title = {Metric spaces},
  year = {2020},
  publisher = {Springer},
  doi = {10.1017/9781139030267.003},
  journal = {An Introduction to Functional Analysis},
}

@book{rudin2021principles,
  author = {Rudin, Walter},
  title = {Principles of mathematical analysis},
  edition = {3},
  year = {1964},
  publisher = {McGraw-Hill},
  url = {https://books.google.com/books?id=kwqzPAAACAAJ},
  doi = {10.2307/3608793},
}

@book{kaplansky2020set,
  author = {Kaplansky, Irving},
  title = {Set theory and metric spaces},
  volume = {298},
  year = {1972},
  publisher = {American Mathematical Society},
  url = {https://bookstore.ams.org/chel-298},
  doi = {10.2307/2319424},
}

@book{watrous2018theory,
  author = {Watrous, John},
  title = {The theory of quantum information},
  year = {2018},
  publisher = {Cambridge university press},
  doi = {10.1017/9781316848142},
}

@book{rudin1987real,
  author = {Rudin, W.},
  title = {Real and Complex Analysis},
  series = {Mathematics series},
  edition = {3},
  year = {1968},
  publisher = {McGraw-Hill},
  url = {https://books.google.com/books?id=NmW7QgAACAAJ},
  journal = {Mathematical Gazette},
  doi = {10.2307/3611894},
}

@article{gyurik2022towards,
  author = {Gyurik, Casper and Cade, Chris and Dunjko, Vedran},
  title = {Towards quantum advantage via topological data analysis},
  journal = {Quantum},
  volume = {6},
  pages = {855},
  year = {2020},
  publisher = {Verein zur F{\"o}rderung des Open Access Publizierens in den Quantenwissenschaften},
  doi = {10.22331/q-2022-11-10-855},
}

\end{document}